\documentclass[journal, onecolumn, draftclsnofoot, 12pt]{IEEEtran}
\usepackage[T1]{fontenc}
\usepackage{url}
\ifCLASSINFOpdf
\else
\fi

\usepackage{amsmath}
\usepackage{xcolor}
\usepackage{algorithmicx}
\usepackage{algorithm}
\usepackage{graphicx}
\usepackage{mathtools}
\usepackage{amssymb}
\usepackage{algpseudocode}
\usepackage{nccmath,textcomp}
\usepackage{eso-pic}
\usepackage{amsthm}
\usepackage{bbm}
\usepackage{verbatim}
\newtheorem{lemma}{Lemma}
\newtheorem{corollary}{Corollary}

\newtheorem{theorem}{Theorem}
\newtheorem{remark}{Remark}

\makeatletter
\let\MYcaption\@makecaption
\makeatother

\usepackage[font=footnotesize]{subcaption}

\makeatletter
\let\@makecaption\MYcaption
\makeatother
\begin{document}
\bstctlcite{IEEEexample:BSTcontrol}
\title{\LARGE Probability of Pilot Interference in Pulsed Radar-Cellular Coexistence: Fundamental Insights on Demodulation and Limited CSI Feedback}
\author{Raghunandan M. Rao,~\IEEEmembership{Student Member,~IEEE,} Vuk Marojevic,~\IEEEmembership{Senior Member,~IEEE,} and Jeffrey H. Reed,~\IEEEmembership{Fellow,~IEEE}
\thanks{Raghunandan M. Rao and Jeffrey H. Reed are with Wireless@VT, Bradley Department of ECE, Virginia Tech, Blacksburg, VA, 24061, USA (email: \{raghumr, reedjh\}@vt.edu).}
\thanks{Vuk Marojevic is with the Department of ECE, Mississippi State University, Mississippi State, MS, 39762, USA (e-mail: vuk.marojevic@ece.msstate.edu).}
\thanks{The support of the U.S. National Science Foundation (NSF) Grants CNS-1564148 and CNS-1642873 are gratefully acknowledged.}}


\maketitle
\vspace{-30pt}
\begin{abstract}
This paper considers an underlay pulsed radar-cellular spectrum sharing scenario, where the cellular system uses pilot-aided demodulation, \textit{statistical} channel state information (S-CSI) estimation and limited feedback schemes. Under a realistic system model, upper and lower bounds are derived on the probability that \textit{at least} a specified number of pilot signals are interfered by a radar pulse train in a \textit{finite CSI estimation window}. Exact probabilities are also derived for important special cases which reveal operational regimes where the lower bound is achieved. Using these results, this paper (a) provides insights on \textit{pilot interference-minimizing} schemes for accurate coherent symbol demodulation, and (b) demonstrates that pilot-aided methods fail to accurately estimate S-CSI of the \textit{pulsed radar interference channel} for a wide range of radar repetition intervals.
\end{abstract}
\vspace{-10pt}

\IEEEpeerreviewmaketitle
\vspace{-8pt}
\section{Introduction}
Over the last decade, radar-cellular spectrum sharing has been actively pursued by academia and industry, due to its high potential in maximizing spectral utilization of heavily congested sub-6 GHz frequency bands. Due to the support for underlay spectrum sharing in radar-incumbent 1.3 GHz \cite{WhiteHouse_1point3_GHz_Consideration}, 3.5 GHz \cite{FCC_3point5_GHz_Rules}, and 5 GHz \cite{FCC_5_GHz_FirstOrder} bands, cellular technologies have progressed from licensed bands to unlicensed and shared bands through standards such as License Assisted Access (LAA) and 5G New Radio-Unlicensed (NR-U) \cite{5G_Evolution_IEEEAccess_2019}. Therefore, evaluating the impact of radar interference on cellular signals is important for network providers. In particular, pulsed radar systems occupying these bands \cite{Rec_ITU_R_M_1465_3} can potentially interfere with control channels of the cellular signal, thus disrupting critical functionalities of the cellular network. 

Pilot/Reference signals are used in modern cellular systems to estimate the instantaneous channel state information (I-CSI), and statistical CSI (S-CSI) of the wireless channel. Due to practical considerations, I-CSI is used at the receiver for channel equalization and coherent demodulation, S-CSI is leveraged at the transmitter to choose the optimal transmission mode for data blocks in subsequent time slots \cite{Qiu_DL_Precode_Stat_CSI_TVT_2018}, \cite{Liu_Geraci_ICSI_HowMany_TWC_2019}. Frequency division duplex (FDD) systems quantize \textit{pilot-aided S-CSI estimates} at the receiver and feed the information back to the transmitter using `limited feedback' schemes \cite{Limited_FB_Love_JSAC_2008}. This methodology is used in NR, where I-CSI is estimated using demodulation reference signals (DMRS), and S-CSI estimates are based on CSI-reference signals (CSI-RS) \cite{3GPPRel15_138_214}.

Pilot interference due to pulsed radar signals impact the accuracy of CSI estimates. It has been demonstrated that pilot-aided I-CSI estimates are corrupted by pilot interference \cite{Karlsson_TDD_Massive_MIMO_Jam_2017}. In contrast, pilot interference is desirable for S-CSI acquisition, since pilot-aided S-CSI estimates account for the interference only when fading and interference statistics are the same on pilot and non-pilot resources. The authors of \cite{Safavi_Roy_ICC_SINR_Nrwbnd_rad_2015} reported degraded turbo decoder performance in the case of pulsed radar-LTE spectrum sharing scenarios due to inaccurate interference estimates, resulting in block decoding failures. The authors of \cite{rao2019analysis} considered the LTE downlink impaired by structured \textit{non-pilot interference}, and demonstrated the inaccuracy of pilot-aided SINR estimates, which resulted in significant degradation of throughput and latency performance due to link adaptation failure.

Unlike conventional multi-cellular scenarios where the interference statistics is homogeneous on all resources, for the radar-cellular coexistence scenarios considered in this letter, the radar is \textit{pulsed and periodic} in nature. Hence, the cellular channel is bimodal with two states: (a) \textit{interference channel}, on data blocks impaired by pulsed radar interference as well as fading, and (b) \textit{fading channel}, on data blocks that are impaired only by fading. While it is desirable to acquire I-CSI using pilots in the fading channel state, it is necessary to acquire S-CSI for \textit{both channel states} to maximize cellular performance using link adaptation and scheduling. 

For robust link adaptation, estimating the S-CSI of the \textit{interference channel} is fundamentally important to maximize performance of the cellular link, as well as to minimize interference to the radar \cite{Liu_Geraci_ICSI_HowMany_TWC_2019}, \cite{rao2019analysis}. However, since pulsed radar interference is time-selective, the absence of pilot interference can result in inaccurate pilot-aided S-CSI estimates of the interference channel.

Before we investigate the effectiveness of pilot-aided S-CSI and I-CSI estimation methods which were not designed for pulsed radar-cellular coexistence, we need to characterize the \textit{probability of pilot interference}. While an exact analysis can be done by considering a finite radar pulse width \cite{Clarkson_Perkins_InterceptTime_TIT_1996}, the resulting expression involving recurrence relations does not facilitate intuitive interpretation. To remedy this, we use a realistic \textit{infinitesimal wideband radar pulse} model that allows us to derive the bounds as a rational function of the waveform parameters, and then prove the achievability of the lower bound. These results lead to important insights regarding the effectiveness of pilot-aided I-CSI/S-CSI estimation, and limited S-CSI feedback.
\subsubsection*{Contributions}
We consider underlay pulsed radar-cellular spectrum sharing, where the radar waveform and cellular pilots are modeled as independent pulse trains with a random initial offset, having different pulse widths and repetition rates. The cellular system employs two different pilot signals, one for S-CSI acquisition and the other for I-CSI estimation. The cellular receiver performs S-CSI estimation using multiple equispaced pilots in a finite \textit{estimation window} and uses \textit{limited S-CSI feedback} to aid in scheduling and link adaptation at the transmitter. Also, the receiver estimates the I-CSI using a different equispaced pilot sequence, for coherent data demodulation  \cite{3GPPRel15_138_214}. Under this model, we derive upper and lower bounds on the probability that a pulsed radar with an infinitesimal pulse width and uniformly distributed time of arrival \cite{Clarkson_Perkins_InterceptTime_TIT_1996} interferes with (a) at least one pilot-bearing OFDM symbol (henceforth referred to as a pilot signal/pilot), and (b) more than $m$ 
pilot symbols, in an \textit{arbitrary estimation window}. We briefly discuss the impact of multipath on the probability bounds, and then derive exact expressions for important special cases where the lower bound is achieved. Using these results, pilot interference-minimizing radar schemes for accurate I-CSI estimation can be obtained. A key insight for pulsed radar-5G coexistence scenarios is that for S-CSI acquisition of the \textit{interference channel}, blind methods need to augment pilot-aided methods for a wide range of radar repetition intervals.

\section{System Model} \label{System_Model}
We consider an underlay radar-cellular spectrum sharing scenario, where the orthogonal frequency division multiplexing (OFDM)-based cellular signal has a symbol duration of $T_\mathtt{ofdm}$. The pulsed radar system has a repetition interval of $T_\mathtt{rep}$, where $T_\mathtt{rep} > T_\mathtt{ofdm}$. Therefore, an OFDM symbol is interfered by at most one radar pulse\footnote{If $T_\mathtt{rep} \leq T_\mathtt{ofdm}$, then each OFDM symbol will be interfered by the radar, and the probability of pilot interference will be $1$.}. Typical high bandwidth radar pulse widths $(T_\mathtt{pulse})$ satisfy  $T_\mathtt{pulse} \ll T_\mathtt{ofdm}$\footnote{In sub-6 GHz systems, typical radar systems have $T_\mathtt{pulse} \sim 1\ \mu$s \cite{Rec_ITU_R_M_1465_3}, while typical values of $T_\mathtt{ofdm} \sim 70\ \mu$s \cite{3GPPRel15_138_214}.}. Hence, we assume that $T_\mathtt{pulse}\rightarrow 0$ and that the radar can be represented by a periodic impulse train, as shown in Fig. \ref{Fig1_Rad_pil}. 

The cellular system employs pilot-aided CSI estimation techniques, where $T_\mathtt{pil}$ denotes  the temporal spacing between pilots. For example, $T_\mathtt{DMRS}$ denotes the DMRS spacing and $T_\mathtt{CSIRS}$ denotes the CSI-RS spacing. 
Even though we focus on \textit{pilot-aided statistical CSI (S-CSI) estimation}, this analysis is general and also applicable to pilot design for optimizing I-CSI (I-CSI) acquisition, as discussed in section \ref{Sec4_Implications}. S-CSI is estimated for each pilot-bearing OFDM symbol in the estimation window of interest denoted by $[0,T_{CSI}]$, where $T_{CSI} = N_p T_\mathtt{pil} = N_\mathtt{ofdm} T_\mathtt{ofdm}$ for $N_p, N_\mathtt{ofdm} \in \mathbb{N}$ and $1 \leq N_p < N_\mathtt{ofdm}$. Here, $N_p$ is the number of pilots, and $N_\mathtt{ofdm}$ the total number of OFDM symbols in the estimation window. The estimated S-CSI using the $l^{th}$ pilot ($\mathtt{CSI}_l$) is mapped to the achievable rate\footnote{LTE and NR define the quantized S-CSI values, how they are fed back, and the S-CSI-to-throughput mapping function $r(\cdot)$ \cite{3GPPRel15_138_214}.} $R_l= r(\mathtt{CSI}_l)$ using a non-zero real-valued function $r(\cdot)$. Defining $\mathbf{R} \triangleq [R_1,R_2,\cdots,R_{N_p}]^T$ as the vector of achievable rates estimated by the receiver, we consider two S-CSI feedback schemes $Q(\mathbf{R})$, given by:
\begin{enumerate}
	\item Minimum S-CSI, calculated using $Q_\mathtt{min}(\mathbf{R}) = \min(\mathbf{R})$,
	\item Window-averaged S-CSI \cite{Qiu_DL_Precode_Stat_CSI_TVT_2018}, calculated using $Q_\mathtt{avg}(\mathbf{R}) = A(\mathbf{R})$, where $A(\cdot)$ is a window-averaging function \cite{Qiu_DL_Precode_Stat_CSI_TVT_2018}.
\end{enumerate} 
As a first-order approximation, a pilot-aided S-CSI estimate of the \textit{interference channel} is considered to be accurate if the pilot is affected by interference. The maximum number of radar pulses that occur in the estimation window is $N_r = \lceil T_{CSI}/T_\mathtt{rep} \rceil$, where $\lceil \cdot \rceil$ denotes the ceiling function. Since typical cellular systems continuously transmit pilot signals for CSI acquisition, and pilot interference is the event of interest, we consider the pilot start and end times to be deterministic. We consider a finite estimation window in which pilot signals occupy the time intervals $[kT_\mathtt{pil}, kT_\mathtt{pil} + T_\mathtt{ofdm}]$ for $k=0,1,\cdots,(N_p-1)$. Due to deterministic pilot intervals, the time of arrival (ToA) of the first radar pulse $t_f$ is assumed to be uniformly distributed, i.e. $t_f \sim \mathtt{U}([0,T_\mathtt{rep}])$ \cite{Clarkson_Perkins_InterceptTime_TIT_1996}. 

\begin{figure}[t]
	\centering
	\includegraphics[width=4.5in]{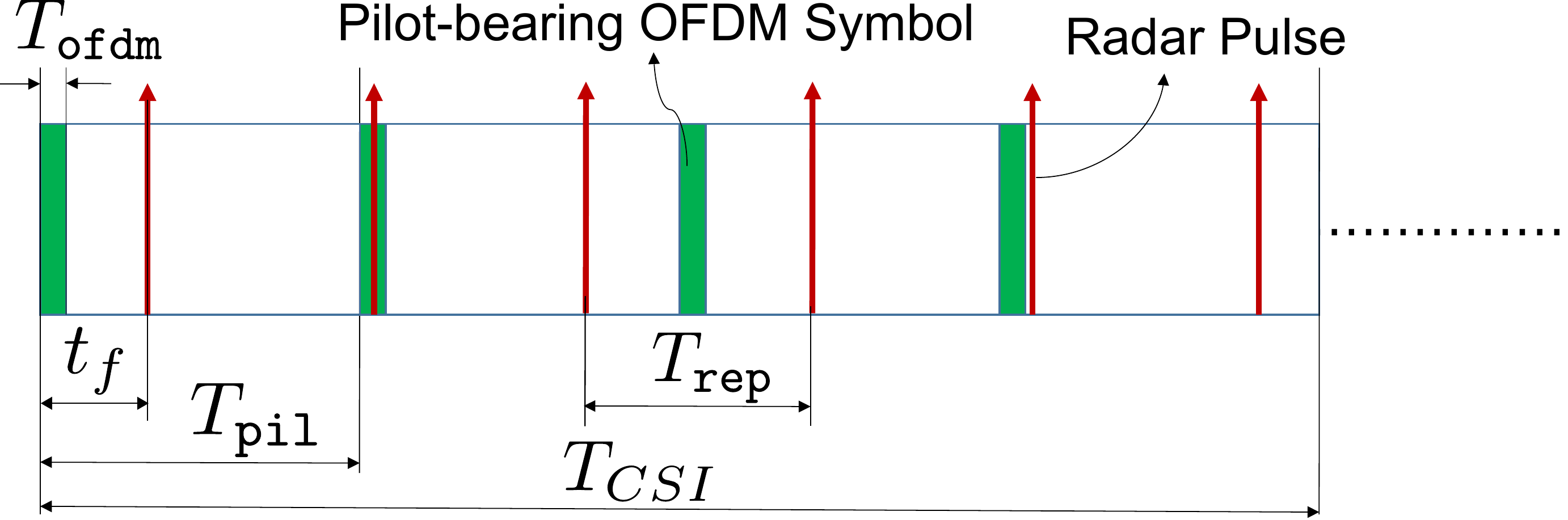}\\
	[-2ex]
	\caption{Illustration of the radar-cellular coexistence scenario. A pulsed radar with repetition interval $T_\mathtt{rep}$ interferes with an OFDM signal with pilots spaced $T_\mathtt{pil}$ seconds apart. Here, the CSI estimation interval ($T_{CSI}$) is comprised of $N_p = 4$ pilot-bearing OFDM symbols.}
	\vspace{-10pt}
	\label{Fig1_Rad_pil}
\end{figure}

\vspace{-10pt}
\section{Probability of Pilot Interference in a Finite CSI Estimation Window}\label{Sec3_PI_radar_pulse}
Let the random variable $M \in \{1,2,\cdots,N_p\}$ denote the number of pilots affected by the pulsed radar signal in the estimation window. In the following analysis, we are interested in the probability that (a) $\{ M \geq 1 \}$, and (b) $\{ M \geq m \}$, for $m=2, 3, \cdots, N_p$. 
\vspace{-15pt}
\subsection{Bounds on Probability of Pilot Interference}
Since $t_f \sim \mathtt{U}[0, T_\mathtt{rep}]$, we have $\mathbb{P}[M \geq 1]$ when $N_r \geq 1$, as shown in the following key result.
\begin{lemma} \label{Lemma_1_UL_bounds}
	If $T_\mathtt{rep} \leq T_{CSI}$, for $m = 2,3,\cdots,N_p$, we have
	\begin{align} 
		\label{eq:P_M_geq 1}
		\frac{T_\mathtt{ofdm}}{T_\mathtt{pil}} & \leq \mathbb{P}[M \geq 1] \leq \min\Big(1, \frac{N_p T_\mathtt{ofdm}}{T_\mathtt{rep}} \Big) \\
		\label{eq:P_M_geq m}
		0 & \leq \mathbb{P}[M \geq m] \leq \min\Big(1, \frac{N_p T_\mathtt{ofdm}}{m T_\mathtt{rep}} \Big).
	\end{align}
\end{lemma}
\begin{proof}
Since $T_\mathtt{rep} \leq T_{CSI}$, $\mathbb{P}[M \geq 1]$ cannot be smaller than the fraction of time allocated to pilots in the CSI estimation interval. Hence, $\mathbb{P}[M \geq 1] \geq \tfrac{N_p T_\mathtt{ofdm}}{T_{CSI}}$. Substituting $T_{CSI} = N_p T_\mathtt{pil}$ and simplifying, we obtain the lower bound.
Similarly, $\mathbb{P}[M \geq 1]$ cannot be greater than the ratio between the total time allocated to the $N_p$ pilots per estimation window and the radar pulse repetition interval. Therefore, $\mathbb{P}[M \geq 1] \leq \min \big(1, \tfrac{N_p T_\mathtt{ofdm}}{T_\mathtt{rep}} \big)$.

For $m$ pilot signals to be interfered by radar in an estimation window $T_{CSI}$, at least one pilot signal must be affected every $\tfrac{T_{CSI}}{m}$ seconds, since both pilots and radar pulses are equispaced in our model. Hence, using the upper bound in Lemma \ref{Lemma_1_UL_bounds} and noting that there are an average of $\tfrac{N_p}{m}$ pilot signals every $\tfrac{T_{CSI}}{m}$ seconds, we obtain the upper bound. 
\end{proof}
Achievability of the lower bounds are discussed next. Even though the lower bound of $\mathbb{P}[M \geq m]$ is zero, it has important consequences on the limited S-CSI feedback of the \textit{interference channel}, as discussed in section \ref{Sec4_Implications}. 
\begin{remark}
	The presence of strong multipath can result in (a) radar pulse broadening if it is due to local scatterers near the user, or (b) interference by multiple echoes in the presence of far-away specular reflections in the channel. For (a), if the broadened radar pulse width is $T_\mathtt{pulse}$, the same approach as Lemma \ref{Lemma_1_UL_bounds} can be used to obtain bounds on the probability of `partial radar interference' on pilots, by replacing $T_\mathtt{ofdm}$ by $(T_\mathtt{ofdm} + T_\mathtt{pulse})$ in (\ref{eq:P_M_geq 1})-(\ref{eq:P_M_geq m}). For (b), if there is one LoS component and $(p-1)$ specular reflectors in the environment, the upper bound becomes $\mathbb{P}[M \geq m] \leq \min\big(1, \tfrac{pN_p T_\mathtt{ofdm}}{m T_\mathtt{rep}} \big), m=1,2,\cdots,N_p$. The lower bound in both cases remain the same as before.
\end{remark}

In the following subsection, we analyze the probability for important special cases.
\vspace{-10pt}
\subsection{Exact Analysis for Important Special Cases}
Let $\mathbbm{1} (l,t_f)$ denote the event that the $l^{th}$ pilot ($l=1,2,\cdots,N_p$) is interfered by a radar pulse in the estimation window $[0,T_{CSI}]$, when the ToA of the first radar pulse is $t_f$. It can be written as
\begin{equation}
\mathbbm{1} (l, t_f) = \begin{cases}
1 \quad \text{if } \exists j=1,2,\cdots,N_r \text{ such that } (t_f + jT_\mathtt{rep}-lT_\mathtt{pil}) \in [0, T_\mathtt{ofdm}] \\
0\quad \text{otherwise}. 
\end{cases}
\end{equation}
We can write the conditional probability of $\{ M \geq m |t_f \}$ ($m=1,2,\cdots,N_p$) as
\begin{align}
\label{ROP_Indicator}
\mathbb{P}[M \geq m|t_f] = \begin{cases}
1 \quad \text{if } \sum\limits_{l=1}^{N_p} \mathbbm{1} (l, t_f) \geq m, \\
0 \quad \text{otherwise}.
\end{cases}
\end{align}
Using (\ref{ROP_Indicator}), $\mathbb{P}[M \geq 1]$ is obtained by marginalizing $t_f$ using $f_{T_f}(t_f)=\tfrac{1}{T_\mathtt{rep}}, 0\leq t_f \leq T_\mathtt{rep}$ to get
\begin{align}
\label{Prob_atlst_1pil_aff_rad}
\mathbb{P}[M \geq 1] & = \int\limits_{0}^{\min(T_{CSI}, T_\mathtt{rep})} \frac{1}{T_\mathtt{rep}} \mathbb{P}[M \geq 1|t_f] d t_f.
\end{align}
The upper limit of the integral is $\min(T_\mathtt{rep}, T_{CSI})$ and accounts for cases where $T_\mathtt{rep} \geq T_{CSI}$, since the observation window of interest is limited to $[0, T_{CSI}]$. 
\begin{theorem} \label{Theorem_1}
The lower bound $\mathbb{P}[M \geq 1] = \frac{T_\mathtt{ofdm}}{T_\mathtt{pil}}$ is obtained for $T_\mathtt{rep} \leq T_{CSI}$ if $T_\mathtt{rep} = k T_\mathtt{pil}$, where $k \in \{1,2,\cdots,N_p \}$. 
\end{theorem}
\begin{proof}
If $T_\mathtt{rep}=kT_\mathtt{pil}$ and $k\in \mathbb{N}$, we have the following \textit{mutually exclusive events}:
\begin{enumerate}
\item $\mathcal{E}_0$: If no pilot in $[0, T_\mathtt{rep}]$ is interfered by the radar, then no pilot will ever be interfered. In other words, $\mathbb{P}[M \geq 1|\mathcal{E}_0] = 0$.
\item $\mathcal{E}_1$: If the $l^{th}$ pilot is affected by radar, then the $(l+mk)^{th}$ pilot will be interfered $\forall \ m \in \mathbb{Z}$. Therefore,  $\mathbb{P}[M \geq 1|t_f] = 1$ for $t_f \in [lT_\mathtt{pil} , lT_\mathtt{pil} + T_\mathtt{ofdm}]$ where $l=0,1,\cdots, (k-1)$. 
\end{enumerate}
Applying the total probability theorem in (\ref{Prob_atlst_1pil_aff_rad}), using $T_\mathtt{rep} = kT_\mathtt{pil}$ in the above and simplifying, we obtain the desired result.
\end{proof} 
The exact value of $\mathbb{P}[M \geq 1] \text{ for } T_\mathtt{rep} \geq T_{CSI}$ is provided in the following corollary.
\begin{corollary}\label{Corollary_2}
$\mathbb{P}[M \geq 1] = \tfrac{N_p T_\mathtt{ofdm}}{T_\mathtt{rep}}$ for $T_\mathtt{rep} \geq T_{CSI}$. 
\end{corollary}
\begin{proof}
The proof is similar to Theorem \ref{Theorem_1}, obtained by direct substitution of (\ref{ROP_Indicator}) in (\ref{Prob_atlst_1pil_aff_rad}).
\end{proof}
\begin{figure}[t]
	\centering
	\includegraphics[width=5.0in]{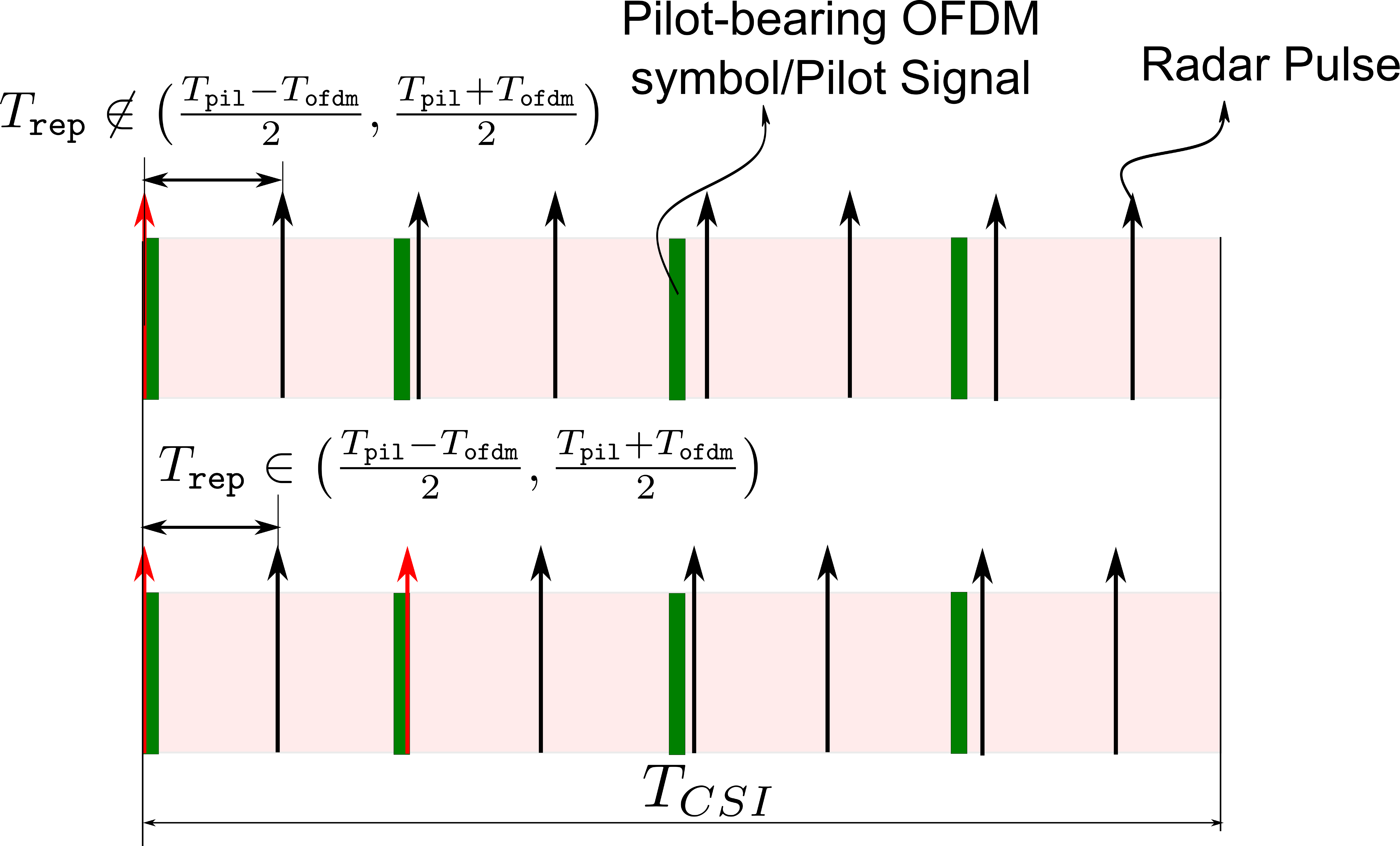}\\
	[-2ex]
	\caption{Illustration of Theorem \ref{Theorem_2_range_of_Trep}, for $N_p = 4, k=1$ and $q=2$ when $m=2$ pilots are interfered by radar pulses (indicated in red).}
	\vspace{-10pt}
	\label{Fig_Proof_Lemma_5}
\end{figure}

Finally, we derive the set of $T_\mathtt{rep}$ for which $\mathbb{P}[M \geq m] > 0$.
\begin{theorem} \label{Theorem_2_range_of_Trep}
For $N_p >1$ and $m=2,\cdots,N_p$, 
\begin{align}
	\label{P_M_geq_m_iff_proof}
	\mathbb{P}[M \geq m] & = \begin{cases}
	\text{non-zero} & \text{if } T_\mathtt{rep} \in \mathcal{T}_{m, N_p} \\
	0 & \text{if } T_\mathtt{rep} \notin \mathcal{T}_{m, N_p} \cap (T_\mathtt{ofdm}, \infty),
	\end{cases} \nonumber \\
	\text{where } \mathcal{T}_{m, N_p} & = \bigcup_{k \in \mathcal{K}, q \in \mathbb{N}} \big(\tfrac{(m-1)kT_\mathtt{pil} - T_\mathtt{ofdm}}{(m-1)q}, \tfrac{(m-1)kT_\mathtt{pil} +T_\mathtt{ofdm}}{(m-1)q} \big), \text{and } \mathcal{K} & = \big \{1,2 ,\cdots, \big\lceil \tfrac{N_p - 1}{m-1} \big\rceil \big \}.
\end{align}

\end{theorem}
\begin{proof}
Due to equispaced pilots, it can be deduced using Theorem \ref{Theorem_1} (specifically, event $\mathcal{E}_1$) that multiple pilots are interfered when $T_\mathtt{rep}$ is \textit{in some neighborhood} of $kT_\mathtt{pil}$, where  $k\in \mathbb{N}$. To interfere with \textit{at least} $m$ pilots in the CSI estimation window, one of the neighborhoods can be shown to be $\big(\tfrac{-T_\mathtt{ofdm}}{m-1}, \tfrac{T_\mathtt{ofdm}}{m-1} \big)$ using the following construction: Conditioned on the event that the first pilot is interfered, there exists some $t_f \in [0, T_\mathtt{ofdm}]$ for which the subsequent $(m-1)$ radar pulses interfere with a pilot if $T_\mathtt{rep} \in \big(kT_\mathtt{pil} - \tfrac{T_\mathtt{ofdm}}{m-1}, kT_\mathtt{pil} + \tfrac{T_\mathtt{ofdm}}{m-1} \big)$. The lower and upper limits of this interval correspond to $t_f = T_\mathtt{ofdm}$ and $t_f = 0$ respectively.
In addition, $k$ must satisfy $k \leq \big \lceil \tfrac{N_p-1}{m-1} \big \rceil$ to ensure that at least $m$ radar pulses are present in $[0, (N_p - 1)T_\mathtt{pil} + T_\mathtt{ofdm}]$ for $\{\mathbb{P}[M \geq m] > 0 \}$ to be true. Since $\big(kT_\mathtt{pil} - \tfrac{T_\mathtt{ofdm}}{m}, kT_\mathtt{pil} + \tfrac{T_\mathtt{ofdm}}{m} \big) \subset \big(kT_\mathtt{pil} - \tfrac{T_\mathtt{ofdm}}{m-1}, kT_\mathtt{pil} + \tfrac{T_\mathtt{ofdm}}{m-1} \big)$ for all $m > 1$, \textit{at least} $m$ pilots are interfered by the radar in the estimation window if $T_\mathtt{rep} \in \mathcal{T}^{(\mathtt{1})}_{m, N_p} = \bigcup_{k \in \mathcal{K}} \big(kT_\mathtt{pil} - \tfrac{T_\mathtt{ofdm}}{m-1}, kT_\mathtt{pil} + \tfrac{T_\mathtt{ofdm}}{m-1} \big)$. In addition, $T_\mathtt{rep} \in \mathcal{T}^{(q)}_{m, N_p} = \{ \tfrac{T}{q} \big| T \in \mathcal{T}^{(1)}_m, q \in \mathbb{N} \}$ can also result in non-zero $\mathbb{P}[M \geq m]$, since $T_\mathtt{rep} \in \mathcal{T}^{(1)}_{m, N_p}$ scaled down by an integer factor preserves the time offset relationship between the radar pulse train and the pilots, as shown in Fig. \ref{Fig_Proof_Lemma_5}. Therefore, $\mathbb{P} [M \geq m]$ \textit{is non-zero} if $T_\mathtt{rep} \in \bigcup_{q \in \mathbb{N}} \mathcal{T}^{(q)}_{m, N_p} = \mathcal{T}_{m, N_p}$. 

Furthermore, if $\mathbbm{1}(1,t_f) = \cdots = \mathbbm{1}(j,t_f) = 0$ and $\mathbbm{1}(j+1, t_f) = 1$ for $j=1, \cdots, (N_p - 1)$, it can be seen that $\mathbb{P} [M \geq m] > 0$ if $T_\mathtt{rep} \in \mathcal{T}_{m, N_p - j} \subset \mathcal{T}_{m, N_p}$, using a similar construction. 

Finally, we notice that $\mathbb{P}[M \geq m] = 1\ \forall\ m=1,2,\cdots, N_p$ if $T_\mathtt{rep} \in [0, T_\mathtt{ofdm}]$, since every OFDM symbol will be interfered in this case. Since \textit{all the feasible} $T_\mathtt{rep}$ values which ensure that $P[M \geq m]$ is non-zero are contained in $\mathcal{T}_{m, N_p}$, we have $\mathbb{P}[M \geq m] = 0$ if $T_\mathtt{rep} \notin \mathcal{T}_{m, N_p} \cap (T_\mathtt{ofdm}, \infty)$. 
\end{proof}
Before we discuss the implications of these results on I-CSI estimation and S-CSI feedback, we validate their accuracy using numerical results.
\begin{figure}[t]
	\centering
	\begin{subfigure}[t]{0.48\textwidth}
		\centering
		\label{Fig3a}
		\includegraphics[width=3.4in]{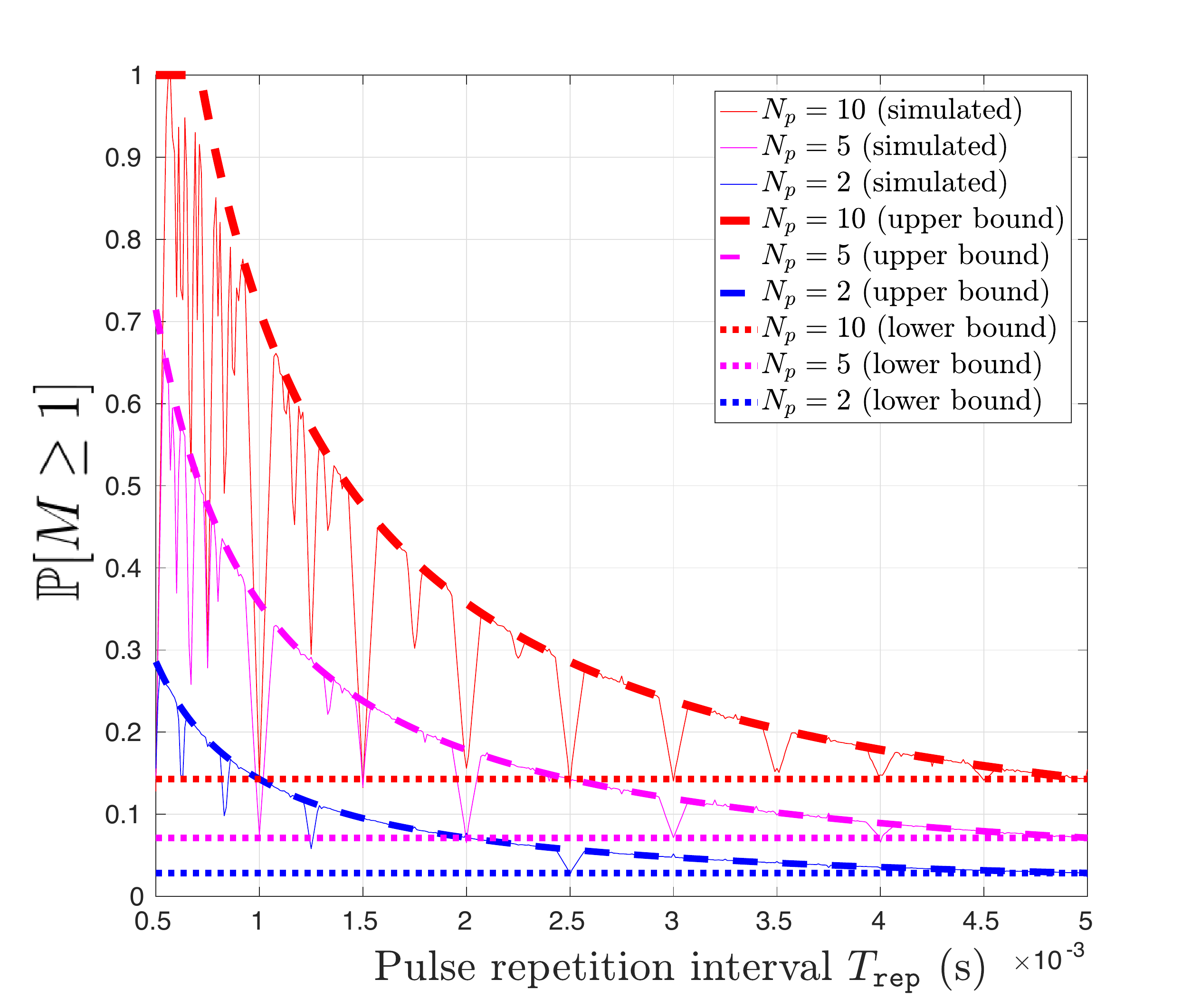}\\
		[-1ex]
		\caption{}
	\end{subfigure}
	~ 
	\begin{subfigure}[t]{0.48\textwidth}
		\centering
		\label{Fig3b}
		\includegraphics[width=3.4in]{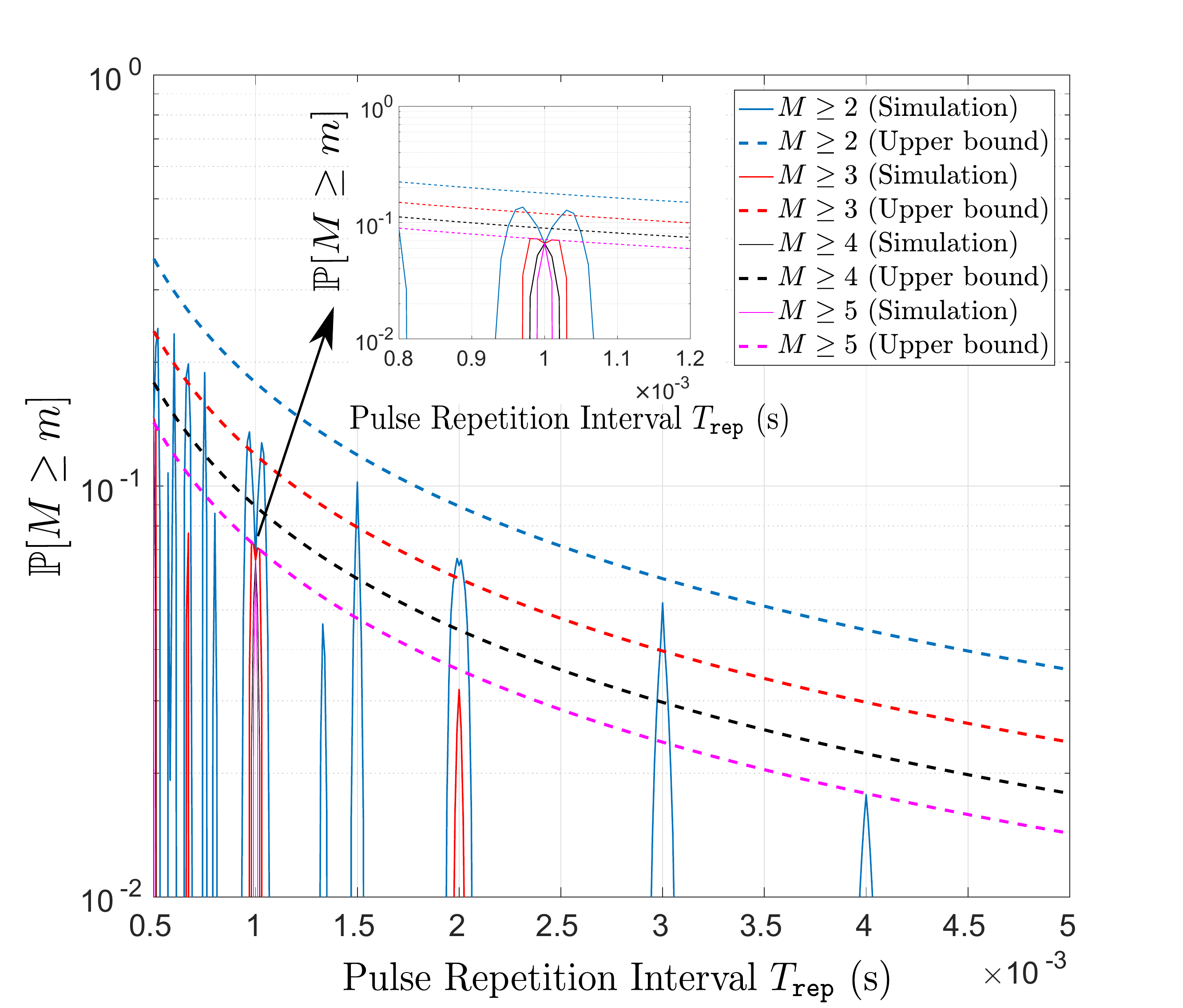}\\
		[-1ex]
		\caption{}
	\end{subfigure}
	\caption{(a) $\mathbb{P}[M \geq 1]$, and its upper and lower bounds as a function of $T_\mathtt{rep}$ and $N_p$ for $T_\mathtt{ofdm} = 71.43\ \mu$s and $T_\mathtt{CSI} = 5$ ms. (b) $\mathbb{P}[M \geq m]$ for $m>1$, and its upper bound as a function of $m$, for $N_p = 5$. $\mathbb{P}[M \geq m] > 0$ only when $T_\mathtt{rep}$ lies in a small neighborhood of a rational fraction of $T_\mathtt{pil}$.}
	\vspace{-10pt}
	\label{Fig3_SimResults}
\end{figure}
\vspace{-10pt}
\subsection{Numerical Results}
We consider a cellular system with a typical OFDM symbol duration of $T_\mathtt{ofdm}=71.43\ \mu s$, and $N_p \in \mathbb{N}$ periodically spaced pilot-bearing OFDM symbols per estimation window of length $T_{CSI}=5$ ms. Fig. \ref{Fig3_SimResults}(a) shows the values of $\mathbb{P}[M\geq 1]$, along with the corresponding upper and lower bounds for different values of $T_\mathtt{rep}$ and $N_p$. We observe that the upper and lower bounds derived in Lemma \ref{Lemma_1_UL_bounds} are in agreement with the numerical results. Furthermore, we also observe that the lower bound is achieved for $T_\mathtt{rep}=kT_\mathtt{pil}, k \in \mathbb{N}$, as proven in Theorem \ref{Theorem_1}.

Fig. \ref{Fig3_SimResults}(b) shows the variation of $\mathbb{P}[M \geq m]$, for $m=1,2,\cdots,5$ in an estimation window of length $T_{CSI}=5$ ms. We observe that the upper bound in (\ref{eq:P_M_geq m}) is in agreement with the numerical results. More importantly, it can be seen that $\mathbb{P}[M\geq m] > 0$ iff $T_\mathtt{rep} \in \mathcal{T}_{m,5}$, as proven in Theorem \ref{Theorem_2_range_of_Trep}. 
\vspace{-7pt}
\section{Fundamental Insights on Coherent Demodulation and Limited S-CSI Feedback}\label{Sec4_Implications}
In this section, we derive new insights on pilot-aided demodulation and limited S-CSI feedback.
\vspace{-10pt}
\subsection{Minimizing Impact on Coherent Demodulation}
It is well known that corrupted I-CSI is detrimental to coherent demodulation \cite{Karlsson_TDD_Massive_MIMO_Jam_2017}. Therefore, minimizing $\mathbb{P}[M \geq 1]$ for DMRS over an infinite observation interval ($T_{CSI} \rightarrow \infty$) \textit{on average} minimizes the occurrence of pulsed radar-induced I-CSI contamination. Using Theorem \ref{Theorem_1}, the lower bound of $\mathbb{P}[M \geq 1]$ is achieved if $T_\mathtt{rep} = kT_\mathtt{DMRS}$ for finite $k \in \mathbb{N}$. Therefore, DMRS interference can be minimized as follows.

\subsubsection{Partial Radar-Cellular Cooperation}
If partial radar-cellular cooperation is feasible, the radar can adapt $T_\mathtt{rep}$ based on (a) prior knowledge, or (b) explicit feedback of $T_\mathtt{DMRS}$.

\subsubsection{Absence of Radar-Cellular Cooperation}
In fading channels with slowly varying channel statistics, throughput can be enhanced by adapting the pilot spacing in time and frequency in real-time, as a function of the channel conditions \cite{Raghu_TVT_2017}. In addition, we propose minimizing $\mathbb{P}[M \geq 1 | T_\mathtt{rep}]$ for DMRS over an infinite observation interval ($T_{CSI} \rightarrow \infty$) \textit{on average} to mitigate I-CSI contamination. Mathematically, the optimal DMRS spacing ($T_\mathtt{DMRS,opt}$) is obtained using
\begin{align} 
\label{Optimiz_problem}
T_\mathtt{DMRS,opt} & = \underset{T_\mathtt{DMRS} \in \mathbb{R}^+}{\text{arg min}}\ \mathbb{P}[M \geq 1|T_\mathtt{rep}], \nonumber \\
\text{s.t. } & T_\mathtt{DMRS} \leq T_\mathtt{coh}.
\end{align}
The constraint is introduced to ensure accurate channel estimation, whereby the DMRS spacing should be smaller than the coherence time $(T_\mathtt{coh})$ \cite{Hassibi_Tcoh_TIT_2002}. In general, an exact solution cannot be obtained due to the aforementioned constraint. 

Nevertheless, a heuristic solution can be obtained using Theorem \ref{Theorem_1} by observing that local minima occur at $T_\mathtt{DMRS}=T_\mathtt{rep}/k,k\in \mathbb{N}$, where $\mathbb{P}[M \geq 1|T_\mathtt{rep}] = \tfrac{kT_\mathtt{ofdm}}{T_\mathtt{rep}}$. The best case scenario occurs when $k=1$, and $T_\mathtt{DMRS,opt} = T_\mathtt{rep}$.  In order to satisfy the constraint, the pilot spacing can be chosen as $T_\mathtt{DMRS} = \tfrac{T_\mathtt{rep}}{k_\mathtt{opt}}$, where $k_\mathtt{opt} = \big\lceil \tfrac{T_\mathtt{rep}}{T_\mathtt{coh}} \big\rceil$. 
To perform this adaptation in real-time, $\hat{T}_\mathtt{rep}$ should be estimated, especially in the case of military radar systems where $T_\mathtt{rep}$ is often unknown. 
\vspace{-10pt}
\subsection{Impact on Limited S-CSI Feedback of the Interference Channel}
Pilot-aided S-CSI estimates of the \textit{interference channel} is inaccurate if pilots are impaired with low probability, or not impacted at all \cite{Safavi_Roy_ICC_SINR_Nrwbnd_rad_2015}, \cite{rao2019analysis}. Under our system model, (a) $\mathbb{P}[M \geq 1]$ is equivalent to the probability that S-CSI of the \textit{interference channel} is accurately acquired using $Q_\mathtt{min}(\mathbf{R})$, and (b) $\mathbb{P}[M \geq m]$ denotes the probability S-CSI of the \textit{interference channel} is accurately acquired using $Q_\mathtt{avg}(\mathbf{R})$.

In contrast to I-CSI, limited feedback of $Q_\mathtt{min}(\mathbf{R})$ is inaccurate for the interference channel when $T_\mathtt{rep} = k T_\mathtt{CSIRS}, k \in \mathbb{N}$.  
Furthermore, (a) upper bound of the probability of obtaining $m$ accurate S-CSI estimates of the \textit{interference channel state} decreases with $m$, (Lemma \ref{Lemma_1_UL_bounds}), and (b) this probability \textit{is non-zero} only when $T_\mathtt{rep} \in \mathcal{T}_{m, N_p}$ (Theorem \ref{Theorem_2_range_of_Trep}). Both of these results imply that window-averaged S-CSI $(Q_\mathtt{avg}(\mathbf{R}) )$ is not reliable for S-CSI acquisition of the interference channel state, since $\mathbb{P}[M \geq N_p/2]=0$ for a large range of $T_\mathtt{rep}$ values. As a result, link adaptation and scheduling schemes in 5G NR will be inefficient in the presence of high-power radar pulses with low $T_\mathtt{rep}$, since pilot-aided schemes fail to capture S-CSI of the \textit{interference channel}. Therefore, \textit{blind S-CSI estimation methods} need to be used to augment pilot-aided estimates, when sharing spectrum with pulsed radars.

\begin{figure}[t]
	\centering
	\includegraphics[width=4.0in]{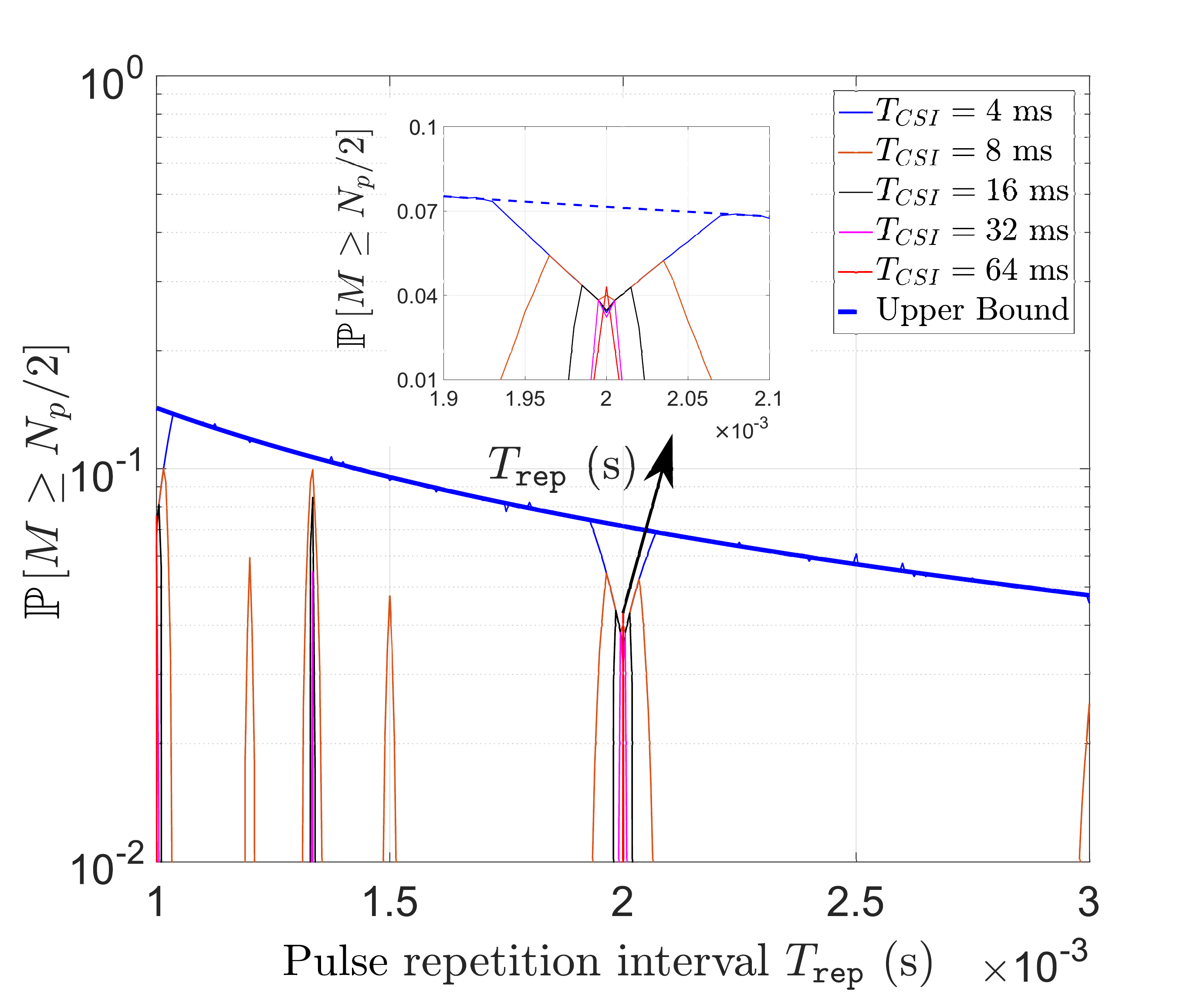}\\
	[-2ex]
	\caption{Plot of $\mathbb{P} \big[M \geq \tfrac{N_p}{2} \big]$ and its upper bound as a function of $T_\mathtt{rep}$, for fixed pilot spacing $T_\mathtt{CSIRS}=2$ ms and estimation window lengths $T_{CSI}=4,8,16,\cdots,64$ ms.}
	\vspace{-10pt}
	\label{Fig4_Impact_of_TCSI}
\end{figure}
\vspace{-5pt}
\subsection{Numerical Results}
Fig. \ref{Fig4_Impact_of_TCSI} shows the impact of the estimation window length on the mean/median-based limited S-CSI feedback schemes in 5G NR, for $T_\mathtt{ofdm} = 71.43\ \mu$s, $T_{CSI} = 4,8,\cdots, 64$ ms \cite{3GPPRel15_138_214}, and $2 \text{ ms} \leq T_\mathtt{rep} \leq 3$ ms. We observe that for a fixed pilot spacing of $T_\mathtt{CSIRS} = 2$ ms, the upper bound of $\mathbb{P}[M \geq \tfrac{N_p}{2}]$ is the same for all cases. However, we also observe that increasing $T_{CSI}$ shrinks the set of $T_\mathtt{rep}$ values for which mean/median S-CSI will be accurate for the \textit{interference channel}. This behavior can be explained using Theorem \ref{Theorem_2_range_of_Trep}: Since $\mathcal{T}_{N_p/2, N_p} = \bigcup_{k \in \{1,2\}, q \in \mathbb{N}} \big(kT_\mathtt{CSIRS} - \tfrac{T_\mathtt{ofdm}}{(N_p/2-1)q}, kT_\mathtt{CSIRS} + \tfrac{T_\mathtt{ofdm}}{(N_p/2-1)q} \big)$, increasing $T_{CSI}$ while keeping $T_\mathtt{CSRIS}$ constant increases $N_p$, thus contracting the size of $\mathcal{T}_{N_p/2, N_p}$. Therefore, increasing the estimation window length while keeping the pilot spacing fixed degrades the availability of accurate S-CSI estimates for the \textit{interference channel state}, when mean or median S-CSI feedback is used. In particular, sparsity of CSI-RS in the time domain \cite{3GPPRel15_138_214} reduces the effectiveness of pilot-aided S-CSI estimation and limited feedback schemes in pulsed radar-NR spectrum sharing scenarios.
\vspace{-5pt}
\section{Conclusion}
Considering an underlay pulsed radar-cellular spectrum sharing scenario, we derived bounds on the probability of single and multiple pilot-bearing OFDM symbols being interfered in a finite estimation window. We proved achievability of the lower bound, and provided insights on designing \textit{pilot interference-minimizing schemes} as a function of the pilot spacing and the radar repetition interval. We also proved that the probability of multiple cellular pilots being  interfered by radar pulses in the estimation window is zero for a large set of radar repetition intervals. This is detrimental for pilot-aided statistical CSI estimation in the \textit{interference channel}, which highlights the need for blind methods in NR and beyond-5G systems sharing spectrum with radars. We demonstrated the accuracy of the derived expressions, and usefulness of the design principles using examples from 5G NR. As cellular networks evolve beyond 5G, these results and insights will be crucial for demodulation reference signal design and robust S-CSI acquisition and feedback schemes. This work can be extended to analyze these probabilities in the case of a pulse radar with an arbitrary staggering sequence, and coexistence between MIMO pulsed radars and MIMO communication systems. A practical application of our work is to study the impact of pulsed interference power on the throughput and latency performance resulting from inaccurate S-CSI estimates.
\bibliographystyle{IEEEtran}
\bibliography{RadarNPI_references}
\ifCLASSOPTIONcaptionsoff
  \newpage
\fi

\end{document}